\documentclass{llncs}
\usepackage{wrapfig}

\usepackage{fancyhdr,graphicx,amsmath,amssymb}
\usepackage[ruled,vlined]{algorithm2e}
\include{pythonlisting}

\usepackage{color}
\usepackage{paralist}
\usepackage{semantic}
\usepackage{xspace}
\usepackage[caption=false]{subfig}
\usepackage{xcolor}
\usepackage{listings}
\usepackage[T1]{fontenc}
\usepackage{minibox}
\usepackage[latin1]{inputenc}
\usepackage{environ}
\usepackage[]{algorithm2e}
\usepackage{algpseudocode}

\usepackage{array}

\IfFileExists{.usepdfcomments}
{\usepackage[icon=Comment,color=notecolor,hoffset=-5pt,voffset=8pt]{pdfcomment}
}
\usepackage{setspace}

\newtheorem{mydef}{Definition}
\newtheorem{myprob}{Problem}
\newcommand{\pred}{\textit{pred}} 
\newcommand{\lhs}{LHS} 
\newcommand{\rhs}{RHS} 
\newcommand{\static}{static} 

\usepackage{url}

\newcommand{\tool}{\textsc{Swapper}}

\newcommand{\trainset}{Training set}
\newcommand{\testset}{Testing set}

\newcommand{\Sk}{\textsc{Sketch}~}
\newcommand{\sk}{\textsc{Sketch}}
\newcommand{\SEARCH}{\textsc{Search}}
\newcommand{\TEST}{\textsc{Test}}
\newcommand{\TRAIN}{\textsc{Train}}

\newcommand{\figlabel}[1]{\label{f:#1}}
\newcommand{\seclabel}[1]{\label{s:#1}}
\newcommand{\figref}[1]{Fig.~\ref{f:#1}}     

\newcommand{\soopt}{Auto-generated}
\newcommand{\cpopt}{Baseline}

\newcommand{\rewriter}{\textbf{simplifier}}
\newcommand{\performanceFunction}{\textbf{fopt}}

\begin{document}

\title{Automatic Generation of Formula Simplifiers based on Conditional Rewrite Rules}
\titlerunning{Synthesizing Formula Simplifiers}  
%
\author{Rohit Singh, Armando Solar-Lezama}
\authorrunning{Singh et al.} 
%
%
\institute{Massachusetts Institute of Technology\\
\email{rohitsingh@csail.mit.edu} \\ \email{asolar@csail.mit.edu}}

\maketitle              

\begin{abstract}
This paper addresses the problem of creating simplifiers for logic formulas based on conditional term rewriting. In particular, the paper focuses on a program synthesis application where formula simplifications have been shown to have a significant impact. We show that by combining machine learning techniques with constraint-based synthesis, it is possible to synthesize a formula simplifier fully automatically from a corpus of representative problems, making it possible to create formula simplifiers tailored to specific problem domains. We demonstrate the benefits of our approach for synthesis benchmarks from the SyGuS competition and automated grading.
\keywords{program synthesis, formula rewriting}
\end{abstract}

 \section{Introduction}
\label{sec:intro}

Formula simplification plays a key role in SMT solvers and solver-based tools. SMT solvers, for example, often use local rewrite rules to reduce the size and complexity of the problem before it is solved through some combination of abstraction refinement and theory reasoning~\cite{z3,boolector}. 
Moreover, many applications that rely on solvers often implement their own formula simplification layer to rewrite formulas before passing them to the solver~\cite{bitvectorrewrites,jeeves,bbr,sparklangvcsmt,muz}. 

In some cases, the motivation to build a simplifier in the application is to be able to reduce the formula on the fly as it is being constructed, rather than constructing a large formula that then gets simplified by the solver. A more fundamental motivation, however, is that formulas generated by a particular tool often exhibit patterns that can be easily exploited by a custom formula simplifier but which would not be worthwhile to exploit in a general solver.

One domain where this simplification is particularly important is constraint-based synthesis. In particular, this paper will focus on the  \sk{} synthesis system~\cite{sketch}. \sk{} is an open source tool which has been used in systems published in multiple top-tier conferences~\cite{parsketch,jsketch,minimize,story,sqlsk,autograder,DBLP:conf/pldi/Solar-LezamaJB08}, and it includes a publicly available benchmark suite of synthesis problems from several important domains including storyboard programming~\cite{story}, query extraction~\cite{sqlsk} and automated grading (Autograder)~\cite{autograder}, sketching Java programs~\cite{jsketch},  as well as 
benchmarks available from the annual synthesis competition (SyGuS)~\cite{sygus} . In particular, \sk{} has hundreds of benchmarks for Autograder problems obtained from student submissions to introductory programming assignments on the edX platform. 

The \sk{} synthesis system includes a formula simplifier, that has a tremendous impact on the size of the formulas. Without these simplifications, problems that solve in seconds time-out instead. 
The simplification module in \sk{} has been under continuous development for eight years, so at this point the code is quite mature and efficient, but it took significant effort to get there, and many hard-to-identify bugs were introduced in the process of developing it, as evidenced from the commit logs in the project repository~\cite{sketchbend}.

In this paper, we present a methodology for automatically generating a formula simplifier from a corpus of benchmark problems. We also present \tool{}, a framework implements this methodology. 
The input to \tool{} is a corpus of formulas from problems in a particular domain. Given this corpus, \tool{} generates a formula simplifier tailored to the common recurring patterns observed in this corpus.  

The system operates in four phases.  In the first phase, \tool{} uses representative sampling to identify common repeating sub-terms in the different formulas in the corpus. In the second phase, these repeating sub-terms are passed to the rule synthesizer which generates conditional simplifications that can be applied to these sub-terms when certain local conditions are satisfied. These conditional simplifications are the simplification rules which in the third phase must be compiled to actual C++ code (the simplifier) that will implement these simplifications. Not all the simplification rules discovered in the third phase will actually improve the run-times for solving the formula. For example, some may reduce the formula size without making it easier to solve, and some may actually prevent other rules from being applied. In order to identify an optimal subset of rules to incorporate into the final simplifier, the system relies on a fourth phase of auto-tuning that evaluates combinations of rules based on their empirical performance on a subset of the corpus (\trainset{}).  The resulting synthesizer will be shown to work effectively not just for problems in the corpus, but also for other problems from the same domain.

Overall, our work makes the following contributions:
\begin{enumerate}
\item 
We demonstrate how to automate the process of generating conditional rewrite rules specific to the common recurring patterns in formulas from a given domain.
\item We demonstrate the use of autotuning to select an optimal subset of rules and generate an efficient simplifier.
\item We evaluate our approach on multiple domains from the \sk{} synthesis tool and show that the generated simplifiers improve the performance (running times) of \sk{} tool by 30-50\% than the existing hand-crafted simplifier in \sk{}.
\end{enumerate}

 \section{Related Work}
A pre-processing step in constraint solvers and solver-based tools (like Z3, Boolector~\cite{boolector}, \sk{} etc) is an essential one and term rewriting has been extensively used as a part this pre-processing step ~\cite{z3prez,bitvectorrewrites,jeeves,bbr,sparklangvcsmt,muz}. These pre-processing steps are very important and can have a significant impact on performance. As others have observed and our experiments confirm,
however, rewrites that help for one class of problems may not be effective
for problems for a different domain. The task of developing domain-specific simplification rules is time consuming and error-prone, and our technique aims to automate the discovery of such rules.


Each part of our framework solves an independent problem and is different from the state of the art, specialized for our purposes. A recent paper introducing Alive~\cite{llvmrr}, a domain specific language for specifying, verifying, and compiling peephole optimizations in LLVM is the closest to our framework as a whole. Their rewrite rules are guarded by a predicate, they use static analyses to find the validity of those guards, they verify the rules and then compile them to efficient C++ code for rewriting LLVM code: all similar to our phases. However, their system is targeted towards the compilers community and relies upon the developers to discover and specify rewrite rules. Our work is targeted towards the solver community and automatically synthesizes the rewrite rules from benchmark problems of a given domain.

Identification of recurrent sub-graphs from benchmark DAGs is similar in essence to the Motif discovery problem \cite{motif} which is famous because of its application in DNA fingerprinting\cite{motifbook}. This is a very active area of research and recently we have seen some attempts to use sampling \cite{samplingmotif}, machine learning \cite{svmotif} and distributed algorithms \cite{motifscale,BaldwinCLSLV04,SeejaAJ09} to compute the Motifs (statistically significant recurrent sub-graphs) as quickly as possible. Our DAGs, on the other hand, have labeled nodes (labeled with operation types) and our motifs have to account for the fact that some opperations are commutative but not others, which makes direct translation to Motif discovery problem more difficult.

In the superoptimization community, people explore all possible equivalent programs and find the most optimal one. It would not make sense to do that for formulas. But \cite{peephole} came up with this idea of packaging the superoptimization into multiple rewrite rules similar to what we are doing here. Although it looks similar in spirit to our work, there are a few differences. Most importantly, \cite{peephole} uses enumeration of potential candidates for optimized versions of instruction sequences and then checks if it is indeed the most optimal version. Whereas, we use a hybrid approach that primarily relies on constraint based synthesis for generating the rules, which offers a possibility of specifying a structured grammar for the functions. 

The third phase in \tool{} automatically generates simplifier's code (representing an abstract reduction system) is similar to a term or graph rewrite system like Stratego/XT \cite{stratego,LV97} or GrGEN.NET \cite{grgen}. Stratego/XT is a framework for the development of program transformation systems and GrGEN.NET is a Graph Rewrite Generator. They offer declarative languages for graph modeling, pattern matching, and rewriting. Both of these tools generate efficient code for program/graph transformation based on rule control logic provided by the user. 
We build upon their ideas and develop our own compiler because we already had special input graphs (DAGs with operations as labels of nodes) and an existing framework for simplification (\sk's DAG simplification class). Our strategy can be compared with LALR parser generation \cite{LALR} where the next look-ahead symbol helps decide which rule to use. In our case as well we keep around a set of rules that are potentially applicable based on what the algorithm has seen.
  \section{Problem Overview}

In order to describe the details of our approach, we first describe a few key features of the \sk{} synthesis system which will serve both as a component and as a target for \tool{}. \sk{} is an open source system for synthesis from partial programs. A partial program (also called a sketch) is a program with holes, where the potential code to complete the holes is drawn from a finite set of candidates, often defined as a set of expressions or a grammar. Given a partial program with a set of assertions, the \sk{} synthesizer finds a completion for the holes that satisfies the assertions for all inputs from a given input space. 

\sk{} uses symbolic execution to derive a predicate $P(x,c)$ that encodes the requirement that given a choice $c$ for how to complete the program, the program should be correct under all inputs $x$.
The predicate $P(x,c)$  is represented internally as a Directed Acyclic Graph. The DAG represents an expression tree where equivalent nodes have been collapsed into a single node by structural hashing~\cite{hashingref}. The grammar for this language of expressions is shown in \figref{dag_example}. The formula simplification pass that is the subject of this paper is applied to this predicate $P(x,c)$ as it is constructed and before the predicate is solved in an abstraction refinement loop based on counterexample guided inductive synthesis (CEGIS). 

\vspace{-20pt}
\begin{figure}[ht]
\centering
\footnotesize
\[
\vspace{-10pt}
\begin{array}{rlr}
e = &
boolop(e_1, e_2) & \mbox{apply a boolean operator on } e_1, e_2 \\
& arithop(e_1, e_2) & \mbox{apply an arithmetic operator on } e_1, e_2 \\
& src(id) & \mbox{universally quantified input}\\
& ctrl(id) & \mbox{existentially quantified control}\\
& arr_r(e_i,e_a) & \mbox{read at index $e_i$ in array $e_a$}\\
& arr_w(e_i,e_a,e_v) & \mbox{write at index $e_i$ value $e_v$ in array $e_a$}\\
& arr_{create}(e_0,...,e_i) & \mbox{create a new array }\\
& n(c) & \mbox{integer/Boolean constant with value $c$}\\
& mux(e_c, e_0, \ldots e_i) & \mbox{multiplexer chooses based on value $e_c$}\\
& assert(e) & \mbox{assertion of a boolean expression $e$}\\
%
\end{array}
\]
\caption{The language of predicate expression in \sk.}
\figlabel{dag_example}
\end{figure}
\vspace{-10pt}

\subsection{Formula Simplification using Conditional Rewrite Rules}
\label{sec:rr_main}

A rewrite rule  is a pattern matching rule that indicates that when the guard predicate in the center is satisfied, one can replace the pattern on the left with the pattern on the right without changing the semantics of the overall formula.
For example, a simple replacement rule is the following:
\[
	mux(t, c, d) \xrightarrow{~~~c=d~~~}  c
\]
This rule indicates that a term  $mux$ that chooses between $c$ and $d$ based on a condition $t$ can be replaced with the term $c$ as long as $c$ and $d$ can be proved to be always equal, which will be true, for example, if $c$ and $d$ are identical subexpressions---this is easy to detect in the DAG representation as the subexpressions will be represented by the same node because of structural hashing.

The predicate may involve equality, as in the example above, or inequalities; for example,
\[
	or(lt(a, b) , lt(a, d)) \xrightarrow{~~~b<d~~~}  lt(a, d)
\]
Note that in order to apply the rule, the system must be able to guarantee that $b<d$ holds for all possible values of the inputs to the sub-formula. In the \Sk{} simplifier, this is done by running a simple dataflow analysis over the graph to infer ranges of values for every term~\cite{sketchthesis}. 

\begin{mydef}
\label{subsec:rr_def}
A \emph{Conditional Rewrite Rule} is a triple $(\lhs(x),\pred(x),\rhs(x))$ where $x$ is a vector of input variables and $\lhs$, $\rhs$ are expressions that include variables in $x$ as free variables and $\pred$ is a guard predicate defined over the same input variables and drawn from a restricted grammar. The triple satisfies the following formula: 
$$ \forall x. \pred(x) \implies \left( \lhs(x) = \rhs(x)\right)$$
We use the notation 
$
	\lhs \xrightarrow{~~~pred~~~}  \rhs
$
to denote the rewrite rules.
\end{mydef}

The definition of the rule is symmetric with respect to $\lhs$ and $\rhs$, but in practice the rule is going to be used by a system that searches for sub-expressions that match $\lhs$ and replaces them with the expression $\rhs$ if it can prove that the condition $\pred$ is guaranteed to be satisfied. For this reason, our system will search for rules where $\rhs$ is smaller than $\lhs$ with the hope that this will help reduce the size of the overall formula. 

There are two points about these rules that are important to highlight at this point. First, there is a trade-off between the strength of the predicate and the reduction that can be achieved by a rule: rules with weak predicates are easier to match than rules with strong predicates, but rules with strong predicates can offer more aggressive simplification. The second point, however, is that the rules that give the most aggressive size reduction are not necessarily the best ones; for example, a rule may replace a very large $\lhs$ expression with a small $\rhs$ but in doing so it may prevent other rules from being applied, resulting in a formula that is larger than what we would have gotten if the rule had not been applied at all. For these reasons, writing optimal simplifiers based on rewrite rules is a challenging task even for human experts, which motivates our approach of using synthesis and empirical autotuning methods to automatically discover optimal sets of conditional rewrite rules. 

 \section{Pattern Finding: Random Sampling based Clustering}
\label{chap:clustering}

We use a random sampling based clustering method to find ``recurrent'' pattern sub-graphs (Definition~\ref{def:pattern}) from a set of benchmark DAGs. We formalize the intuitive requirement of ``domain significance'' (Definition ~\ref{def:significant}), present our algorithm (Algorithm~\ref{alg:sampling} in the Appendix) that asymptotically finds ``domain significant'' patterns and prove it's correctness.

\begin{mydef}[Pattern, Sub-pattern, Sub-graph]
\label{def:pattern}
A pattern in the context of \tool{} is a Directed Acyclic Graph (DAG) rooted at a node that represent a computation of a single value with each node being labeled by an SMT operation, constant or a typed variable. A sub-pattern or a sub-graph $P$ of a larger DAG $G$ is a pattern such that there exists a graph isomorphism between $P$ and $G$ that preserves the computation labels.     
\end{mydef}

\begin{mydef}[Domain Significance]
\label{def:significant}
Given a set $S$ of benchmark DAGs from a domain $D$, a sub-pattern DAG $P$ of size $N$ (number of operation nodes in the DAG) is said to be domain significant with probability $p$ where $p$ is the ratio of number of occurrences of the pattern $P$ in $S$ and the total number of patterns of size $N$ in $S$. The larger the value of $p$, more domain significant the pattern $P$ would be.     
\end{mydef}


\subsection{Sampling Algorithm}

To motivate the algorithm, we look at the requirement of approximating domain significance:to approximate the distribution given by the actual domain significance probabilities asymptotically (Definition ~\ref{def:significant}). This will be guaranteed if we ensure that every sample of size $N$ is equally likely to be one of all patterns of size $N$ in the benchmark DAGs (allowing repetitions). Because the benchmark set $S$ comprises of DAGs, we can always pick a node at random from $S$ and then try to find a pattern of size $N$ rooted at that node. If we can guarantee that for any fixed node in $S$, our sampling method at that node picks a pattern of size $N$ with the same probability across all patterns of size $N$ rooted at that node, and, that probability is independent of the chosen node then we satisfy the original requirement. Our algorithm maintains a boundary of collected nodes at any stage and picks another node from all available parents (including some NULL parents) randomly at each step. We present the algorithm in the appendix (Algorithm ~\ref{alg:sampling}) and sketch a proof of it's correctness in Theorem~\ref{th:sampling} below.

\begin{theorem}\label{th:sampling}
Given a set $S$ of benchmark DAGs and a number $N \geq 2$, Algorithm~\ref{alg:sampling} asymptotically approximates the underlying order of patterns of size $N$ based on the distribution of the domain significance probabilities.
\end{theorem}

\begin{proof}
As discussed above, it's sufficient to show that the algorithm samples every pattern of size $N$ rooted at a particular node with equal probability irrespective of what the node is. 

For simplicity, let's first assume that instead of labeled DAGs representing formulas, we have labeled directed binary trees in $S$. In this case, such a sampling method is easy to construct. We randomly pick a node from $S$ and then randomly pick from its (two) parents and at any stage pick randomly from all available parents until we have $N$ nodes. If we hit a node with no parents, we restart the sampling. Due to the fact that each node has equal number of parents, at step $i$ (first step being the choice between two parents), the number of available choices for picking will be $i+1$. This ensures that the probability of picking any pattern (possibly after some restarts) of size $N \geq 2$ is the same across the board i.e. is $\frac{1}{2}\times\frac{1}{3}\times...\times\frac{1}{N}$. This idea can be easily extended to $k$-ary trees. 

Moreover, for trees with variable but bounded in-degrees we incorporate ``shadow'' edges (represented by $NULL$ in the algorithm) to make sure that each node has the same number of parents (the bound on in-degrees) and picking a ``shadow'' edge will cause a restart of the sampling process. This will ensure that if we pick a pattern of size $N$ (again, possibly after some restarts), it will be equally likely to be any pattern of size $N$ from $S$. 

Now, for DAGs instead of trees, where edges can point to the same node again, we ensure this property by first picking a node from $S$ at random and then considering only the breadth first search tree rooted at the first chosen node for sampling - this ensures that for every DAG pattern rooted at this node, there is only one way to pick it and hence equal probability of it being picked as compared to any pattern rooted at that node. Once we find the nodes of the pattern, we consider all the edges that were already there for analysis, construct the pattern and count it appropriately in the samples. 

Implementation of the algorithm incorporates the BFS tree construction on the fly by maintaining a list of already seen nodes ($boundayNodes \cup nodes$) and treating any new edge that leads to them as being a ``shadow'' edge (NULL in the pseudo-code). It then extends the list of possible choices to $size(nodes)\times(T-1) + 1$ which is the number of valid possibilities one will have if this were a $T$-ary tree. The probability for any pattern $P$ of size $N \geq 2$ rooted at the first picked node ($n_0$ in the pseudo-code) to be picked by the sampling algorithm (given that a pattern has been picked, potentially after some restarts) is $\frac{1}{T}\times\frac{1}{2*T-1}\times...\times\frac{1}{(N-1)*T-(N-2)}$ because each node will be picked randomly from a list ($boundaryNodes$) of size $i\times(T-1) + 1$ at step $i$ for $1 \leq i \leq N-1$. Note that due to the implicit BFS tree construction, there is only one way for a pattern to be picked and the probability expression is independent of the pattern or the node picked in the beginning. 

\end{proof}


\subsection{Similarity Metrics}
To group DAGs together into clusters, \tool{} uses the following metrics.
\textbf{(1) Signature of the DAG}: To identify equivalent patterns, \tool{} builds a signature string of the DAG which is an encoding of the DAG with parenthesis establishing the types of nodes, ordered names of input nodes and the parent-child relationships. On top of this, the ordering of parents for commutative operations is made lexicographic to ensure that patterns are mapped to the same signature when they are equivalent because of commutativity. 
\textbf{(2) Static Analysis information from the benchmark DAGs}: To check the $\pred$ part of the rules, the generated simplifier performs a static analysis on the DAG and identifies assumptions that are valid for each node. The assumptions for input or source nodes are maintained with the DAG signature. Note that each DAG can have multiple assumptions and all those assumptions are used later used during the rule generation phase but are ignored for aggregating patterns.

\subsection{Stopping Criterion}
\tool{} chooses a large number $M$ (starting with $50000$) but continues sampling until the total number of patterns with probability of occurrence greater than a threshold $\epsilon$ converges i.e. the next $M$ samples do not result into a significant change in the number of such patterns. In our experiments, we used $\epsilon=0.02$.

\section{Rule Generation: Syntax-Guided Synthesis}
\seclabel{rulegen}
In this phase, \tool{} finds corresponding rewrite rules for the set of common patterns obtained from the Pattern Finding phase.
For each pattern, the Pattern Finding phase also collects a 
 set of predicates that can be valid for the inputs of each pattern at different occurrences, as computed by a static analysis on each benchmark in the corpus.


\subsection{Problem Formulation} 

\tool{} needs to find correct rewrite rules (\ref{sec:rr_main}) which have a given \lhs{}. Additionally, we want to avoid rules with predicates that will never hold in practice, so we focus on rules with predicates that are implied by the predicates found by static analysis during the Pattern Finding Phase. 

\begin{myprob}\label{prob:sketch}
Given a function \lhs(x), assumptions \static(x) discovered by static analysis for a given occurrence of the \lhs{} pattern,
and grammars for \pred(x) and \rhs(x), find suitable candidates for \pred(x) and \rhs(x) which satisfy the following constraints:

\begin{enumerate}
\item $\forall x: \static(x) \implies \pred(x)$
\item $
\forall x: \text{ } if\text{ }(\pred(x))\text{ } then
	\text{ }  (\lhs(x) == \rhs(x))
$
\item $size($\rhs$)$ < size($\lhs$), where size($expr$) is the number of terms  in the expression $expr$
\item $\pred(x)$ is the weakest predicate (most permissive) from the predicate grammar that enables the $\lhs$ to $\rhs$ transformation
\end{enumerate}
\end{myprob}

\subsubsection{The space of predicates}
For \pred{} \tool{} employs a simple Boolean expression generator that considers conjunctions of equalities and inequalities among variables.
These predicates are inspired by existing predicates present in the rules in the Sketch's  hand-crafted simplifier and are easier to check statically than more complicated predicates. 

\subsubsection{Grammar for \rhs}

The template for \rhs{} simulates the computation of a function using temporary variables. This computation can be naturally interpreted as a DAG in topologically sorted order. Essential grammar for the generator for \rhs{} is shown here:
\[
\begin{array}{rl}
\rhs(x) \equiv & let ~t_1 = \textbf{simpleOp}(x); \ldots t_{k} = \textbf{simpleOp}(x,(t_1,...,t_{k-1}));\\
~&in ~t_{k}
\end{array}
\]
where \textbf{simpleOp} represents a single operation node (e.g. AND, OR etc) with its inputs being selected from the arguments. We put a strict upper bound on $k$ as the number of nodes in the \lhs{} for which the \rhs{} is being searched for. The last computed temporary is interpreted as the output node of the DAG.

\subsection{Correctness Constraint} 

Setting apart constraint number 4 regarding optimality of $\pred$ in the problem definition \ref{prob:sketch}, this problem can be seen as a classic synthesis problem. We will enforce the fourth constraint on top of solutions to this synthesis problem.

$$ \exists c_p c_r \forall x \text{ } \left(\static(x) \implies \pred(x,c_p)\right) \wedge \pred(x,c_p) \implies \left(\lhs(x) = \rhs(x,c_r)\right) $$

where $c_p$ and $c_r$ are the choices the solver needs to make to get a concrete \pred(x) or \rhs(x). More concretely, these are the choices of: (i) when to expand the grammar or when to use a terminal, and (ii) which subset of inputs to choose for a particular operation (equality in \pred{} and \textbf{simpleOp} in \rhs). 

We explored two different techniques for synthesizing such rewrite rules: (1) Symbolic \sk{} based synthesis of rules, and (2) Enumerative search with heuristics. \tool{} uses a hybrid approach to get the best of the both aforementioned techniques: scalability and exhaustiveness. We describe the hybrid technique briefly. 

\subsection{Hybrid Enumerative/\sk{}-based synthesis}
\label{subsec:ruleswithsketch}

\tool{} breaks the SyGus Problem into two parts.
\textbf{(1) Constraints and optimizations on predicates}:  \tool{} uses the enumerative approach, generates all possible candidate predicates from the specification grammar and checks for their validity based on collected assumptions $\static(x)$. It applies various heuristics and optimizations including predicate refinement (\ref{susec:predrefine}) on top of the symbolic synthesis to scale the overall synthesis process while maintaining it's exhaustive nature.
\textbf{(2) Synthesis of $\rhs$}: \tool{} hard-codes a predicate and realizes the $\rhs$ synthesis problem in \sk{} using the \textbf{generator} and \textbf{minimize} features (\cite{skthesis} \cite{minimize}) of the \sk{} language. \tool{} uses generators to recursively define the template for $\rhs$, and \textbf{minimize} to find the smallest possible \rhs(x) for a fixed \pred(x).

  \begin{figure}[!htbp]
   \begin{minipage}{0.55\textwidth}
			\centering

            \includegraphics[width=0.85\linewidth]{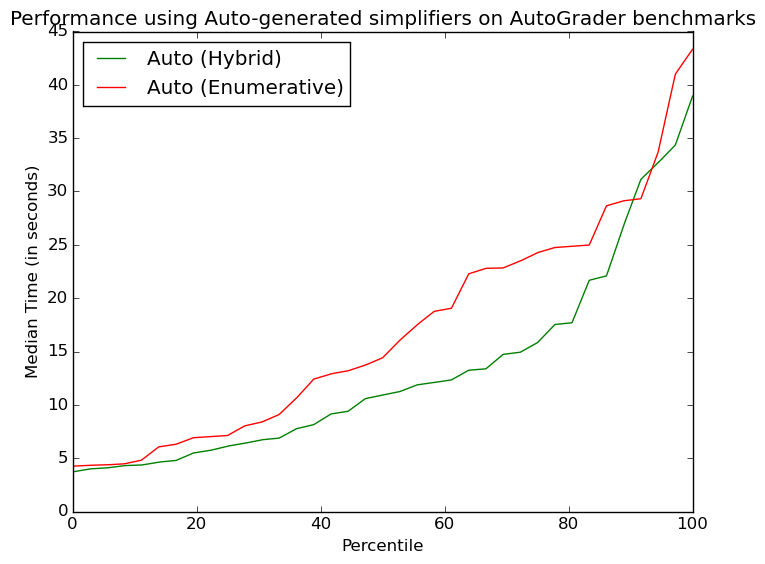}
            \caption{Comparison of performance of \sk{} using the best \soopt{} simplifiers from Hybrid and Enumerative techniques on Autograder benchmarks}
            \label{fig:enumhybrid}

          \end{minipage}
    \begin{minipage}{0.45\textwidth}
    \begin{algorithm}[H]
    \scriptsize
     \SetKwData{Left}{left}\SetKwData{This}{this}\SetKwData{Up}{up}
     \SetKwFunction{Union}{Union}\SetKwFunction{FindCompress}{FindCompress}
     \SetKwInOut{Input}{input}\SetKwInOut{Output}{output}
     \Input{\lhs(x), allPreds}
     \Output{valid rewrite rules of the form
     $(\lhs,\pred_i,\rhs_i)$}

     $G_{\Rightarrow} \leftarrow $ ImplicationGraph(allPreds) \\
     \While{$\neg \text{isEmpty}\left(G_{\Rightarrow}\right)$}{
      $\pred \leftarrow \text{popLeaf}\left(G_{\Rightarrow}\right)$ \\
	 $\rhs \leftarrow \sk\text{Synth}\left(\lhs,\pred\right)$  \\
      \eIf{$\rhs =$ NULL}{
      	\ForEach{$\pred' \in G_{\Rightarrow} \wedge \pred \Rightarrow \pred'$}{
        	removePred$\left(\pred',G_{\Rightarrow}\right)$
        }
       }{
        Output rule $(\lhs,\pred,\rhs)$
      }
     }
     \caption{Refinement method}
     \label{alg:pred_refinement}
    \end{algorithm} 
        \end{minipage}
        
  \end{figure}

\subsubsection{Predicate Refinement}
\label{susec:predrefine}
\tool{} constructs the implication graph ($G_{\Rightarrow}$) of all available predicates and iteratively finds $\rhs$ for the least applicable predicates (leaves in the graph $G_{\Rightarrow}$) at a given stage. Note that this heavily reduces the synthesis time when there's no possible simplification rule for a leaf predicate and \tool{} can prune out all predicates implied by it. The procedure is shown in Algorithm \ref{alg:pred_refinement}.





This hybrid technique has the benefit of being able to exhaustively search for rules of big sizes while making the core synthesis problem faster (fewer constraints) and highly parallelizable (multiple \sk{} instances with different predicates are run in parallel). Figure~\ref{fig:enumhybrid} shows the improvement of the auto-generated simplifier for AutoGrader benchmarks by using the Hybrid technique over just the Enumerative one (excluding \sk). Note that \sk{} does synthesis of rules guaranteeing their correctness for large but bounded inputs, hence, we fully verify the generated rules with z3~\cite{z3} as well by expressing them as SMT constraints before using them for code generation. 


\section{Efficient Simplifier Code Generation }
\label{codegen}

There are no readability constraints to be satisfied by the simplifier code since this is automatically generated by tool{}. This enables \tool{} to do certain optimizations. The most important one being that now it can share the burden of pattern matching among all the rules which is particularly useful when the $\lhs$ is the same for multiple rules. First, using this idea and an abstract predicate evaluation technique \tool{} encodes the process of rule application more efficiently than the earlier optimization phase. Second, to help incorporate all the ``symmetries'' of the rules, \tool{} enables pattern matching with DAGs which differ only in the order of parent nodes at some commutative operation(s). 

\tool{} hard-codes all aspects of the rule application including predicate evaluation and the replacement procedure as C++ code. Because of the engineered efficiency of rewriting, note that, in our experiments we will not consider simplification time because of it being a one-time negligible (always a fraction of a second) time-step as compared to further \sk{} solving.

\section{Auto-tuning Rules}

\tool{} uses OpenTuner~\cite{opentuner} to auto-tune the set of rules according to a performance metric (based on time, memory, size of DAGs etc). OpenTuner uses an ensembles of disparate search techniques simultaneously and quickly builds a model for the behavior of the optimization function treating it as a black box.

\subsection{Optimization Problem Setup}

\tool{} specifies the set of all rules to the tuner and creates the following two configuration parameters:
(1) A permutation parameter: for deciding the order in which the rules will be checked.
(2) Total number of rules to be used.

A lot of rules do not directly interact with each other and hence \tool{} reduces the search space by creating multiple permutation parameters for each set of rules that can interact with each other e.g. any two rules with overlapping patterns can potentially affect each other's applicability whereas any two rules with different base nodes in $\lhs$ will not interact with one another and can therefore be ordered in any way as desired for pattern checking.

The evaluation function (\performanceFunction) to be optimized takes as input a set of rules and returns a real number. This number corresponds to performance improvement of \sk{} on \trainset{} benchmarks after rewriting the benchmarks using the simplifier generated  (Section~\ref{codegen}) from the input set. Intuitively, \performanceFunction{} penalizes heavily for under-performing on any benchmark from the \trainset{} and rewards normally for over-performing. It returns the average of all these rewards. The auto-tuner tries to maximize this reward by trying out various subsets and orderings of rules provided to it as input while learning a model of dependence of \performanceFunction{} on selection of rules.

 \section{Experiments}
\newcommand{\doopt}{Hand-coded}

In order to test the effectiveness of our system, we focus on these three questions: (1) Can \tool{} generate good simplifiers in reasonable amounts of time and with low cost of computational power?
(2) How do the simplifiers generated by \tool{} affect SMT solving performance of \sk{} relative to the hand written simplifier in \sk{}?
(3) How domain specific are the simplifiers generated by \tool{}?

For evaluation of \tool{} on \sk{} domains, we compared the following three simplifiers:
\begin{enumerate}
\item \doopt: This is the default simplifier in \sk{} that has been built over a span of past eight years. It comprises of simplifications based on (a) rewrite rules that can be expressed in our framework(Subsection~\ref{sec:rr_main})  (b) constant propagation (c) structure hashing~\cite{hashingref} and (d) a few other complex simplifications that cannot be expressed in our framework  
\item \cpopt: This disables the rewrite rules that can be expressed in our framework from the \doopt{} simplifier but applies the rest of the simplifications (b)-(d).
\item \soopt: This incorporates the \tool{}'s auto-generated rewrite rules on top of the \cpopt{} simplifier.
\end{enumerate}

In the sections that follow, we elaborate on the details of the experiments.


\subsection{Domains and Benchmarks}
\newcolumntype{L}[1]{>{\raggedright\let\newline\\\arraybackslash\hspace{0pt}}m{#1}}
\newcolumntype{C}[1]{>{\centering\let\newline\\\arraybackslash\hspace{0pt}}m{#1}}
\newcolumntype{R}[1]{>{\raggedleft\let\newline\\\arraybackslash\hspace{0pt}}m{#1}}

\begin{figure}[ht]
\vspace{-0.5cm}
\begin{center}
\scriptsize
\begin{tabular}{ |  L{2.4cm} |  C{4.5cm} |  C{4.5cm} | }
  \hline
  \textbf{Domain} & \textbf{Benchmark DAGs Used} & \textbf{Avg. Number of Terms} \\ 
  \hline
  \textbf{Storyboard} & 34 & 4173 \\
  \hline
  \textbf{QBS} & 14 & 6580 \\
  \hline
  \textbf{AutoGrader} & 45 & 23289 \\
  \hline
  \textbf{Sygus} &  22 &  68366 \\
  \hline
\end{tabular}
\vspace{-0.5cm}
\end{center}
\caption{Domains and Benchmarks investigated for \tool{}'s evaluation}
\vspace{-0.5cm}
\end{figure}

We investigated benchmarks from various domains of \sk{} applications. 
Storyboard benchmarks correspond to specifications for storyboard programming problems~\cite{story}, QBS ones correspond to query extraction specifications~\cite{sqlsk}, Sygus corresponds to the SyGus competition benchmarks translated from SyGus format to Sketch specification~\cite{sygus} and AutoGrader ones are obtained from student's assignment submissions in the introduction to programming online edX course~\cite{autograder}. For each of these domains we picked suitable candidates for \tool{}'s application by (1) eliminating those benchmarks which did not have more than 5000 terms in the formula represented by their DAGs and those which took less than 5 seconds to solve - so that there's enough patterns and opportunity for improvement (2) removing those which took more than 5 minutes to solve - this was done to keep \tool{}'s running time reasonable because we need to run each benchmark multiple times during auto-tuning phase. Using these cutoffs, the total number of usable benchmarks for AutoGrader domain were reduced from 2404 to 45 and for Sygus from 309 to 22. Only a few (3-4) of the QBS and Storyboard benchmarks qualified in this screening process with around 5000-6000 terms each. \tool{} was not able to generate enough rules for auto-tuning from this small set of terms in the benchmarks (\trainset). 





\subsection{Synthesis Time and Costs are Realistic}

To generate a simplifier, \tool{} employs: (1) a single pass of Pattern Finding (2) a single or multiple passes of Rule Generation depending on whether the rules from first pass were satisfactory (3) multiple evaluations made by the Auto-Tuning phase. Many of these computations can be parallelized, especially, the rule generation (we may have up to 100,000 patterns to analyze) and evaluation of benchmarks (Auto-Tuner will potentially call this hundreds of times, this number was around 150 in our experiments).  

We used a private cluster running Openstack as the infrastructure for parallelized computations with parallelisms of 20-40 on two virtual machines emulating 24 cores, 32GB RAM of processing power each. 

We present a worst case estimate of CPU time and cost of computation taken by these phases based on our experiments and the Amazon Web Services~\cite{aws} estimator (using On-demand pricing for the same instances we used on our private cluster and rounding them up to account for volatility in the prices):

\begin{figure}[ht]
\vspace{-0.5cm}
\scriptsize
\begin{center}
\begin{tabular}{ |  L{2.5cm} |  C{1.8cm} |  C{2cm} |  C{1.7cm} |  C{2cm} |  C{1.6cm} | }
  \hline
  \textbf{Domain} & \textbf{Pattern Finding} & \textbf{Rule Generation}  & \textbf{Auto-Tuning} & \textbf{Total Time (in hours)} & \textbf{Cost}\\ 
  \hline
  \textbf{AutoGrader} & 3 hours & $1\text{ hour }\times5$ & $0.08 \times 150$ & 20 & \$22  \\
  \hline
  \textbf{Sygus} & 2 hours & $1\text{ hour }\times5$ & $0.1 \times 150$ & 22 & \$24  \\
  \hline
\end{tabular}
\vspace{-0.5cm}
\end{center}
\caption{Estimated CPU time (in hours) and cost of computation on the cloud}
\vspace{-0.5cm}
\end{figure}

In essence, \tool{} can be used to automatically synthesize a simplifier for less than $\$50$ spent on computation (around what one would pay a good developer for an hour's worth of work). Note that these wait times (around 20-22 hours or a day to get a simplifier) can be improved significantly by: (a) parallelization of the Pattern Finding phase (b) Setting a timeout for evaluation runs that are guaranteed to be worst than existing good runs and (c) increasing the parallelism available to \tool{}.

\subsection{\tool{} Performance}


To test the performance of \tool{} on \sk{} benchmarks from a particular domain, we first divided the corpus into three disjoint sets randomly  $(\SEARCH,\TRAIN,\TEST)$. The \SEARCH{} set was used to find patterns in the domain and \TRAIN{} set was used in the auto-tuning phase for evaluation. And finally, \TEST{} set was used to empirically confirm that the generated simplifier is indeed optimal for the domain. Moreover, we used $2$-fold cross validation to ensure that there was no over-fitting on \TRAIN{} set. We achieved this by exchanging \TRAIN{} and \TEST{} sets and  auto-tuning with the \TEST{} set instead of the \TRAIN{} set. We obtained similar performing simplifiers as a result and verified that there was no over-fitting. 



We implemented the evaluation of benchmarks in \tool{} as a python script that takes a set of DAGs as input, runs \sk{} on each of them multiple times (set to 5 in our experiments) and obtain the median values for time taken or memory consumed. In our experiments we limited the performance evaluation to time because we observed empirically that memory and time are highly correlated on these benchmarks.

We obtained $301$ rules for AutoGrader domain and $105$ rules for Sygus domain. The optimal simplifier for AutoGrader used $135$ of the rules and the one for Sygus used $65$ rules. 

\paragraph{Benefits over the existing simplifier in \sk}
\soopt{} simplifier reduced the size of the problem DAGs by 13.8\% (AutoGrader) and 1.1\%(Sygus) on average as compared to the size of DAGs obtained after running \doopt{} simplifier (Figure~\ref{fig:AG_SY_size}). On DAGs obtained after using \soopt{} simplifier on average, \sk{} solver performed better than on those obtained by using \doopt{} simplifier (Figure~\ref{fig:AG_SY_time}): (1) \soopt{} simplifier made \sk{} run faster on 80\% of the AutoGrader benchmarks and 90\% of the Sygus benchmarks (2) The average times taken by \sk{} to solve a benchmark simplified using \soopt{} simplifier were 13s (AutoGrader) and 8s (Sygus) as compared to 20s and 21s respectively for the \doopt{} simplifier. Figures~\ref{fig:AG_SY_size} and~\ref{fig:AG_SY_time} show distribution of sizes and times for \sk{} solving after applying all three simplifiers with percentiles on the x-axis. It clearly shows the consistent improvement in performance by applying the \soopt{} simplifier.

\begin{figure}[ht]
\centering
\begin{tabular}{cc}
\includegraphics[width=0.4\columnwidth]{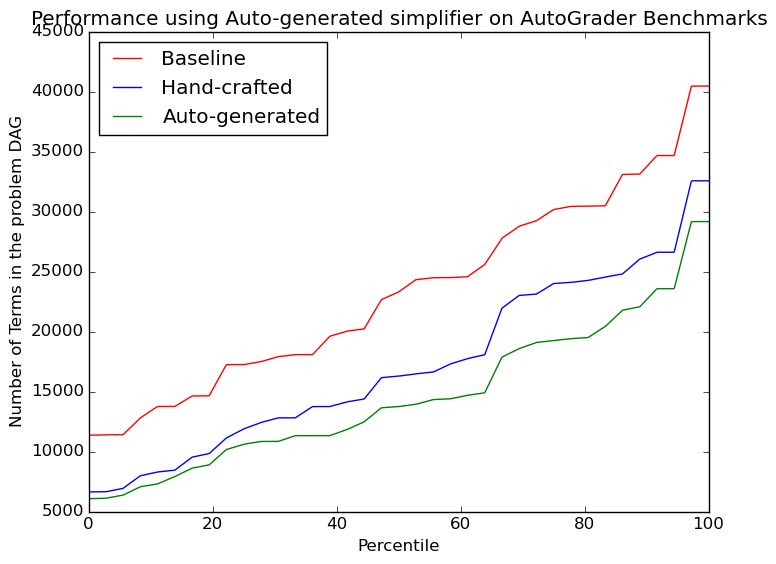} & \includegraphics[width=0.4\columnwidth]{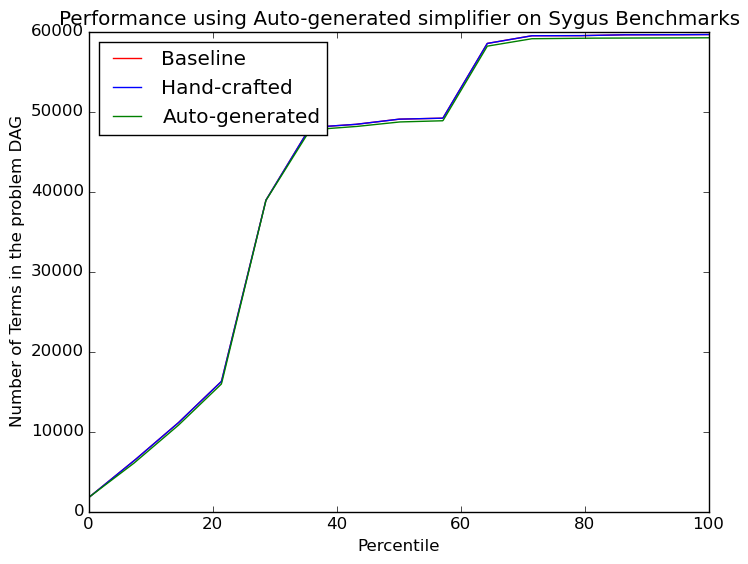}
\end{tabular}
\caption{Change in sizes with different simplifiers}
\label{fig:AG_SY_size}
\vspace{-0.5cm}
\end{figure}
 
 \begin{figure}[ht]
\centering
\begin{tabular}{cc}
\includegraphics[width=0.4\columnwidth]{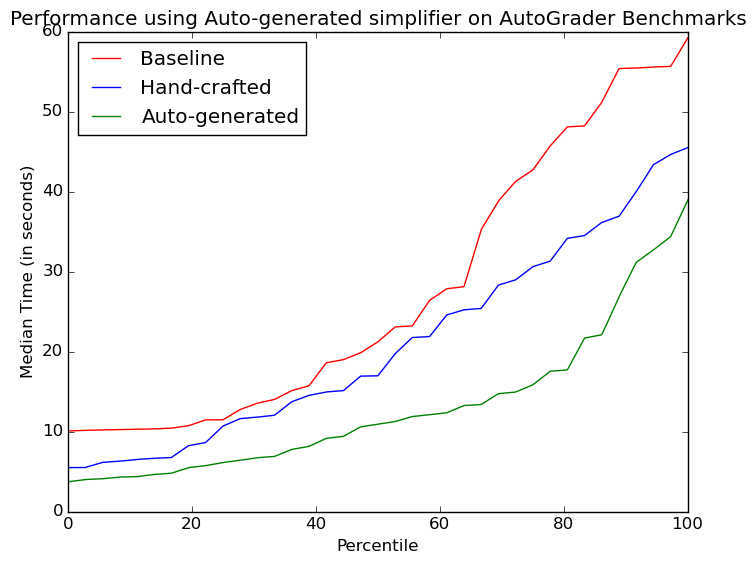} & \includegraphics[width=0.4\columnwidth]{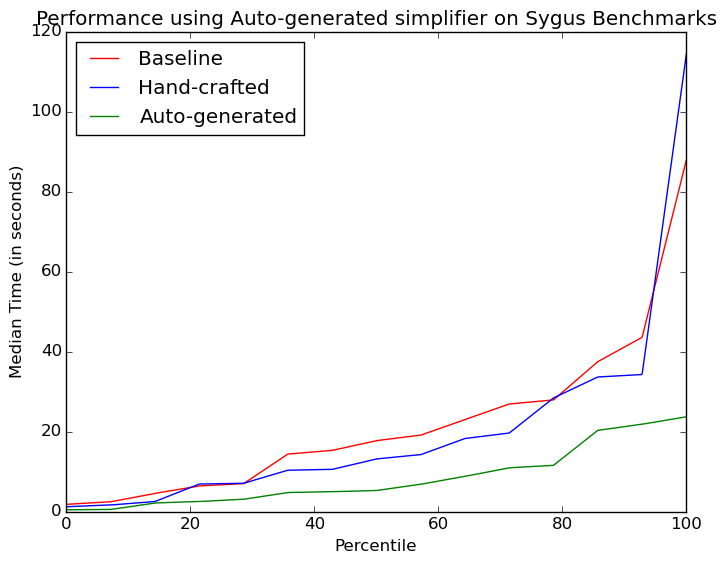}
\end{tabular}
\caption{Median \sk{} running time percentiles with Auto-generated simplifiers}
\label{fig:AG_SY_time}
\vspace{-0.5cm}
\end{figure}
 
\paragraph{Benefits over the unoptimized version of \sk}
\soopt{} simplifier reduced the size of the problem DAGs by 38.6\% (AutoGrader) and 1.6\%(Sygus) on average (Figure~\ref{fig:AG_SY_size}). Application of \soopt{} simplifier results into huge improvements in running times for \sk{} solver on both AutoGrader and Sygus benchmarks as compared to application of \cpopt{}: the average time of solving a benchmark was reduced from  27.5s (AutoGrader) and 22s (Sygus) to 13s and 8s respectively (Figure~\ref{fig:AG_SY_time}). 

 \begin{figure}[ht]
\centering
\begin{tabular}{cc}
\includegraphics[width=0.4\columnwidth]{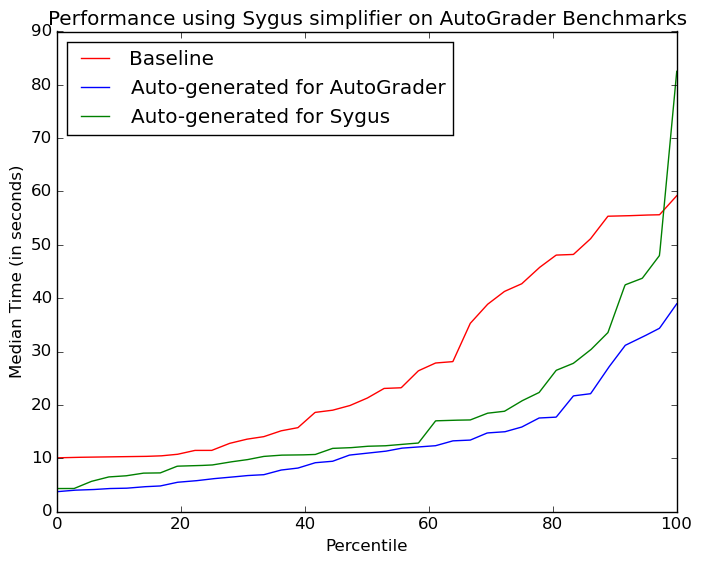} & \includegraphics[width=0.4\columnwidth]{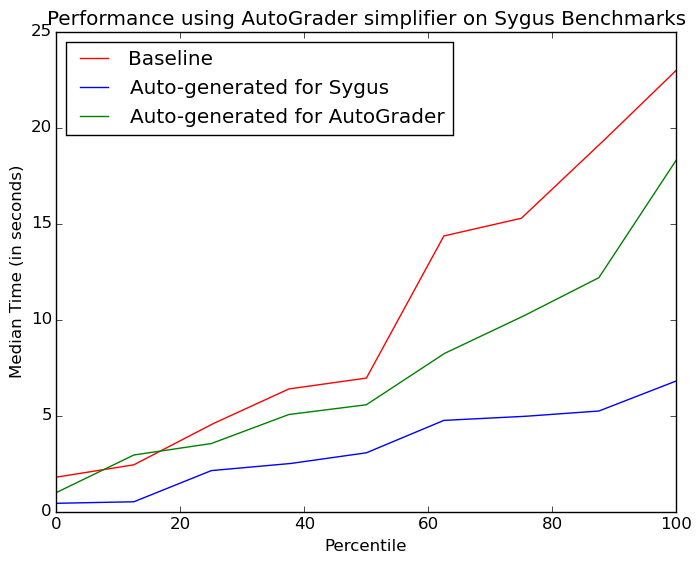}
\end{tabular}
\caption{Domain specificity of \soopt{} simplifiers: Time distribution}
\label{fig:domainspecific}
\vspace{-0.5cm}
\end{figure}

\paragraph{Domain specificity}
We took the \soopt{} simplifier obtained from one domain and used it to simplify benchmarks from the other domain and then ran \sk{} on the simplified benchmarks. Application of \soopt{} simplifier obtained from Sygus increased the \sk{} running times drastically on a few AutoGrader benchmarks when compared to the application of \cpopt{} simplifier (Figure~\ref{fig:domainspecific}), and, resulted into \sk{} running slower than after application of \soopt{} simplifier obtained from AutoGrader domain. Application of \soopt{} simplifier generated for AutoGrader domain reduced the running times of \sk{} solver on average as compared to the \cpopt{} on Sygus benchmarks but the times were still far away from the performance gains obtained by application of \soopt{} simplifier generated for the Sygus domain (Figure~\ref{fig:domainspecific}). This validates our hypothesis of these generated simplifiers to be very domain specific.





\begin{singlespace}
\bibliography{main}
\bibliographystyle{splncs03} 
\end{singlespace}
 \pagebreak
 \section{Appendix}
 

 \subsection{Implementation Details: Summary/Supplemental Materials}
 
\subsubsection{Random Sampling based Clustering}
We built a custom parser for DAGs in Python and implemented the sampling method among other book-keeping operations. This parser takes a corpus of benchmark DAGs as an input and interacts with \sk{} back-end (\cite{sketchbend}, implemented in C++) to  find common patterns and valid local assumptions on the inputs to these patterns. This method is approximate but fairly simple and works well for our purposes in reasonable amounts of time. 

\begin{figure}
\begin{minipage}{\textwidth}
\begin{algorithm}[H]
\scriptsize
\SetKwData{Left}{left}\SetKwData{This}{this}\SetKwData{Up}{up}
\SetKwFunction{Union}{Union}\SetKwFunction{FindCompress}{FindCompress}
\SetKwInOut{Input}{input}\SetKwInOut{Output}{output}
\Input{$S$: Set of benchmark DAGs, ~ $N\geq2$ : Size of pattern to be found \\ $M$: Number of samples to be considered \\ $T$: maximum in-degree of any node in $S$}
\Output{$Samples$: patterns of size $N$ ordered by their approximate domain significance}
\BlankLine
  $i \leftarrow 0$,  $Samples \leftarrow map()$ \\
  \While{i < M}{
    $n_0 \leftarrow $ randomPick($S$) , ~
    nodes $\leftarrow List(n_0)$ \\
    \While{$size($nodes$) < N$}{
      boundaryNodes = $List()$, ~
      boundaryEdges = $List()$\\
      \ForEach{$n\leftarrow $ nodes}{
          \For{$e\leftarrow $ n.incomingEdges}{
              \If{$e.from \in \text{nodes} \cup boundaryNodes$}{
                  boundaryEdges.append(NULL) \\
              }\Else{
                  boundaryNodes.append(e.from) \\
                  boundaryEdges.append(e) 
              }
          }
      }
      \While{size(boundaryEdges) < $size(nodes) \times (T-1) + 1$}{
          boundaryEdges.append(NULL) 
      }
      $e^{*} \leftarrow \text{randomPick}(boundaryEdges)$ \\
      \If{$e^{*} = $ NULL} {
		\textbf{break}
      }
      \Else{
      	nodes.append($e^{*}.from$) \\
      }
    }
    \If{$size(nodes) == N$}{
        $Samples[Pattern(\text{nodes}))] $  += $ 1$ \\
        $i=i+1$
    }
  }
  $sortByValues(Samples)$

\caption{Probabilistic DAG Sampling}\label{alg:sampling}
\end{algorithm}
\end{minipage}
\end{figure}

\textbf{Input:} The input to this phase are benchmark DAG files and the size of the patterns to be found. We use the standard \sk{} back-end DAG specification~\cite{skmanual}. A snippet of an example input file is provided here:
\begin{lstlisting}[language=TeX]
543 = ARR_W INT_ARR 429 536 542
544 = ARR_W INT_ARR 428 543 5
545 = NOT BOOL 482
546 = AND BOOL 361 545
547 = AND BOOL 481 546
\end{lstlisting}

Here the first entry in each line is the node ID followed by the type of operation, the output type of the operation and respective operands (maximum number of operands in our experiments = in-degree of any node in the DAG was 3)

\textbf{Output: } The output from this phase is a list of patterns sorted by their number of occurrence in the samples. We have attached one sample output file for two different domains in the supplemental materials. The list contains the rule signature to identify the rule. Here is an example snippet of such an output:

\begin{lstlisting}[language=TeX]
(AND|(AND|(S:N_3:BOOL)|(S:N_4:BOOL)|)
|(EQ|(S:N_1:INT)|(S:N_2:INT)|)|) 1966 246
(AND|(AND|(S:N_1:BOOL)|(S:N_2:BOOL)|)
|(AND|(S:N_3:BOOL)|(S:N_4:BOOL)|)|) 1702 3
(AND|(AND|(AND|(S:N_3:BOOL)|(S:N_4:BOOL)|)
|(S:N_2:BOOL)|)|(S:N_1:BOOL)|) 1230 4
\end{lstlisting}

The second column is the number of times this pattern occurred in the sampling phase and the third column is the number of different static analyses configurations it found at those locations where the pattern was found. This phase also produces a list of different static analysis configurations for the same domain. For example, for the pattern 

\texttt{(AND|(AND|(S:N\_3:BOOL)|(S:N\_4:BOOL)|)|(EQ|(S:N\_1:INT)|(S:N\_2:INT)|)|)}, 

some static analysis configurations are:
\begin{lstlisting}[language=TeX]
(AND|(AND|(S:N_3:BOOL:R(0-1))|(S:N_4:BOOL:R(0-1))|)
|(EQ|(S:N_1:INT:L(|-1|0|1|2|3|4|))|(S:N_2:INT:L(|-1|
))|)|)
##(AND|(AND|(S:N_3:BOOL:R(0-1))|(S:N_4:BOOL:R(0-1))|)
|(EQ|(S:N_1:INT:L(|-1|0|2|3|4|5|6|7|8|9|10|11|12|13|)
)|(S:N_2:INT:L(|-1|))|)|)
##(AND|(AND|(S:N_3:BOOL:R(0-1))|(S:N_4:BOOL:R(0-1))|)
|(EQ|(S:N_1:INT:L(|-1|0|1|2|3|5|29|30|31|32|34|35|36|
37|44|45|52|53|60|61|68|69|103|104|105|106|113|114|
115|116|118|))|(S:N_2:INT:L(|-1|))|)|)
\end{lstlisting}
Each source or ``S'' node is appended with interval analysis results i.e. it can be in a Range \texttt{R(3-10)} or a in a list \texttt{L(|3|5|7|)} or just be ``anything'' \texttt{T}.

\subsubsection{Rule Generation: Hybrid technique}

We implemented the enumerative part of the technique in Python and parallelized the \sk{} calls using Pathos library~\cite{pathos}. We also augmented the \sk{} back-end to implement some enumerative search heuristics and optimizations for generated rules in C++.

\textit{Optimizations for the Hybrid Technique:}
The synthesis search space is pruned using the following heuristics:
\begin{itemize}
\item Constrain the number of operation nodes in $\rhs$ based on nodes found in $\lhs$ e.g. if $\lhs$ doesn't contain any nodes that write to an array then $\rhs$ shouldn't either or if $\lhs$ contains only Boolean operations then limit the $\rhs$ grammar to Boolean operations as well. This helps reduce \sk{} solving time.  
\item Maintain signatures of each $(\lhs,\pred)$ pair accounting for symmetries so that solving for rules containing equivalent pairs never happens more than once.
\item In \sk{} specification, use the expected data-type of output nodes to pick the final operation in $\rhs$.
\end{itemize}

We attach all generated rules in supplemental materials. The rules are arranged in folders with all files ``*.aux'' corresponding to DAGs representing the LHS (d.aux), RHS (f.aux) and predicate (p.aux).

\begin{lstlisting}[language=TeX]
$ cat ./119/d.aux
0 = S BOOL N_3
1 = S BOOL N_4
2 = S BOOL N_2
3 = S BOOL N_1
4 = OR BOOL 0 1
5 = OR BOOL 4 2
6 = OR BOOL 5 3
7 = AND BOOL 6 1
7
$ cat ./119/f.aux
0 = S BOOL N_3
1 = S BOOL N_4
2 = S BOOL N_2
3 = S BOOL N_1
1

$ cat ./119/p.aux
0 = S BOOL N_3
1 = S BOOL N_4
2 = S BOOL N_2
3 = S BOOL N_1
4 = CONST BIT 1
4
\end{lstlisting}

\subsubsection{Code Generation}
 
We implemented the code generation component in the \sk{} back-end. The Code Generation phase takes the rules provided to it as an input (rule signature strings or separate pattern DAGs in files) and generates C++ code that implements efficient matching of \lhs{} of all the rules together and hard-codes all symmetries of the rules for matching. The generated C++ code is complied again with minimal libraries from \sk{} back-end to produce a binary for the \rewriter. The \rewriter{} is then run on existing set of DAGs to produce reduced simplified DAGs.

We attach auto-generated code in files names ``matches.cpp'' which is compiled with existing \sk{} back-end code.

\subsubsection{Auto-tuning Rules}

We implemented the evaluation function using calls to multiple parallel instances of \sk{} for \testset{} benchmarks in Python. The tuner interacts with a MySQL database to store and analyze existing information about the experiments (run-times, configurations of rules, etc). Note that because

We attach the auto-tuner specification code as python scripts that show the evaluation function.

\subsection{Conclusion}

We conclude by reiterating the contributions made in this paper. We presented a framework \tool{} that automates generation of domain specific simplifiers, for Solvers and other tools employing Solvers, based on conditional rewrite rules. We also tested it for multiple domains in \sk{} showing that \tool{} can generate good simplifiers in reasonable amounts of time and with low cost of computational power. Simplifiers generated by \tool{} perform better than the \doopt{} simplifier on average in \sk{} and are very domain specific i.e. a simplifier obtained for a particular domain will most likely not work well on another different domain. 

%
%

\end{document}